\documentclass{article}
\usepackage[margin=2.5cm]{geometry}
\usepackage{url}
\usepackage{tikz}
\usepackage{hyperref}
\usepackage{caption}
\usepackage{amssymb}

\usepackage{amsmath}
\usepackage{amsthm}
\usepackage{amsfonts}
\usepackage{mathtools}
\usepackage{hyperref}
\usepackage{cleveref}
\usepackage{graphicx}
\usepackage{xcolor}
\usepackage{boxedminipage}
\usepackage{bbm}
\usepackage{algorithm}
\usepackage{algpseudocode}
\usepackage{numbertabbing}

\floatstyle{ruled}
\newfloat{algo}{htbp}{algo}
\floatname{algo}{Algorithm}

\newcommand{\op}[1]{\textsl{#1}}

\newtheorem{definition}{Definition}

\newtheorem{theorem}{Theorem}
\newtheorem{lemma}{Lemma}

\AtBeginDocument{%
  \providecommand\BibTeX{{%
    \normalfont B\kern-0.5em{\scshape i\kern-0.25em b}\kern-0.8em\TeX}}}

\title{Asynchrony-Resilient Sleepy Total-Order Broadcast Protocols}

\author{Francesco D'Amato\\
  Ethereum Foundation\\
  \url{francesco.damato@ethereum.org}
  \and Giuliano Losa\\
  Stellar Development Foundation\\
  \url{giuliano@stellar.org}
  \and Luca Zanolini\\
  Ethereum Foundation\\
  \url{luca.zanolini@ethereum.org}
}
\date{}

\begin{document}
\maketitle
\begin{abstract}\noindent
  Dynamically available total-order broadcast (TOB) protocols tolerate fluctuating participation, e.g., as high as 99\% of their participants going offline, which is especially useful in permissionless blockchain environments.
  However, dynamically available TOB protocols are synchronous protocols, and they lose their safety guarantees during periods of asynchrony.
    This is a major issue in practice.

    In this paper, we propose a simple but effective mechanism for tolerating bounded periods of asynchrony in dynamically available TOB protocols that ensure safety deterministically.
    We propose to trade off assumptions limiting the online/offline churn rate in exchange for tolerating bounded asynchronous periods through the use of a configurable message-expiration period.

    In practice, this allows picking a small synchrony bound~$\delta$, and therefore obtain a fast protocol in the common case, knowing that the protocol tolerates occasional periods of duration at most~$\pi>\delta$ during which the bound does not hold.
    We show how to apply this idea to a state-of-the-art protocol to make it tolerate bounded periods of asynchrony.
\end{abstract}

\section{Introduction}

At the heart of a system like Ethereum is a dynamically available, Byzantine fault-tolerant (BFT) total-order broadcast (TOB) protocol that allows participants to propose new blocks and agree on a growing blockchain.
Byzantine fault-tolerance means that the protocol tolerates attacks in which some participants misbehave, i.e., they are controlled by an adversary and maliciously deviate from the protocol.
Dynamic availability means that the protocol can handle participants going offline or coming back online  at any time --- even 99\% of them --- as long as a sufficient fraction (more than half in Ethereum) of the online participants remain well behaved.

In Ethereum, dynamic availability is crucial to survive a failure scenario in which one or more of Ethereum's consensus clients (which are different implementations of Ethereum's TOB protocol) suffer a software bug and crash, taking down a large fraction\footnote{In May 2023, roughly 60\% of Ethereum's consensus clients went offline for about 25 minutes due to a software bug; Ethereum's dynamically available chain nevertheless continued growing normally.} of the participants at once.
In contrast, traditional BFT protocols (synchronous or partially synchronous) get stuck when participation drops below their fixed (usually 1/2 or 2/3) quorum threshold.

It is unfortunately impossible to achieve consensus, a prerequisite for TOB, in the dynamically-available setting under partial synchrony~\cite{DBLP:conf/fc/LewisPyeR23}\cite[Theorem 7.1]{lewispyePermissionlessConsensus2024}.
Thus, dynamically available TOB protocols are synchronous protocols, meaning that they depend on a known, system-wide upper bound $\delta$ on communication delay and clock skew (assuming local computation takes negligible time).
This allows simulating a round-by-round model (in the vein of the basic round model of Dwork, Lynch, and Stockmeyer~\cite{DBLP:journals/jacm/DworkLS88}), such that, each round, every message sent by an online and well-behaved participant is received by all the participants that are online and well-behaved in the next round.

All existing dynamically available TOB protocols are formulated in variants of this round-by-round model, such as the sleepy model~\cite{sleepy}, that make additional assumptions such as cryptographic assumptions, assumptions about the ratio of ill-behaved to well-behaved online participants, assumptions about eventually stable participation, etc.

Unfortunately, in practice, when the synchrony bound $\delta$ is violated, the message-delivery guarantees of the round-by-round model cannot be guaranteed anymore, and existing dynamically available TOB protocols lose their safety and liveness guarantees.
Therefore, practical deployments must choose an excessively conservative, i.e., large, synchrony bound $\delta$.
This is a problem because the latency and throughput of dynamically available TOB protocols are proportional and inversely proportional, respectively, to $\delta$.

In this work, we propose a methodology to modify existing dynamically available TOB protocols in order to obtain protocols that have the following desirable properties:
\begin{enumerate}
        \item They tolerate asynchronous periods, i.e., periods during which message delivery is under full control of the adversary, of maximum duration $\pi$, which is a parameter; and
        \item They match the latency and throughput of the original protocol when the synchrony bound $\delta$ holds.
\end{enumerate}

In practice, this allows picking a small synchrony bound~$\delta$, and therefore to obtain a fast protocol in the common case, knowing that the protocol tolerates occasional periods of duration at most $\pi>\delta$ during which the synchrony bound does not hold.
With existing dynamically available TOB protocols, maintaining safety under those assumptions would require setting $\delta=\pi$, which would significantly slow down the protocol.

The methodology we present applies to protocols that ensure safety deterministically.
Examples are the protocols of Momose and Ren~\cite{DBLP:conf/ccs/Momose022}, Malkhi, Momose, and Ren~\cite{DBLP:journals/iacr/MalkhiMR22,DBLP:conf/ccs/MalkhiM023}, Gafni and Losa~\cite{gl}, and D'Amato and Zanolini~\cite{DBLP:journals/corr/abs-2310-11331}.
The case of protocols that ensure probabilistic safety is treated separately by D'Amato and Zanolini~\cite{rlmd}.

To build some intuition as to why existing dynamically available TOB protocols lose their safety guarantees when message-delivery guarantees are violated, let us consider the 1/3-resilient, dynamically available TOB protocol of Malkhi, Momose, and Ren~\cite{DBLP:journals/iacr/MalkhiMR22} (the \emph{MMR protocol}, further detailed in Section~\ref{sec:mmr}).
The MMR protocol tolerates less than a fraction $\beta=1/3$ of the online participants being malicious in each round.
The protocol consists of a sequence of views of~2 rounds each.
In each view, participants may introduce new values in the first round of the view and vote to extend the blockchain in the second round, which we call the \emph{decision round}.
Each participant decides on a value~$b$ when more than a fraction $\alpha=1-\beta=\frac{2}{3}$ of the votes it receives in a decision round are for~$b$ (note that this threshold is relative to the number of votes that a participant receives, and not the total number of participants, because there is no telling how many participants are online in a given round).

In order to guarantee progress in the face of changing participation, in each round $r$, each participant only uses votes cast in the same round $r$.
Otherwise, a drop in participation might stall the protocol because the currently-online participants would not be numerous enough, compared to the number of previously online participants, to reach the decision threshold.
For example, if participation drops from 100 to 10 participants from one round~$r$ to the next round~$r+1$, then the votes of the participants that are online in round $r+1$ can obviously not account for~$\frac{2}{3}$ of participants' latest votes over rounds~$r$ and~$r+1$. 

Regrettably, protocols that only use votes from the current round lose all safety guarantees if there are periods of asynchrony during which message-delivery is fully under adversarial control.
For example, suppose that the network delivers only adversarial messages in the decision round of the MMR protocol.
Trivially, if the adversary sends only votes for~$b$ to a participant~$p_i$ and only votes for~$b'\neq b$ to another participant $p_j$, then $p_i$ decides~$v$ because it receives unanimous votes for~$b$ and, similarly, $p_j$ decides~$b'$.
This violates the agreement property of total-order broadcast.

To sum up, on the one hand it seems that, each round, dynamically available protocols like the MMR protocol must only use votes cast in the current round or they lose progress guarantees in case of fluctuation in the participation level from one round to the next.
On the other hand, using only votes cast in the current round means losing safety in asynchronous rounds.
In this paper, we offer a solution to this conundrum: We observe that we can use the most recent votes that each participant casts over a fixed number of previous rounds, called the \emph{expiration period}, without losing safety or progress guarantees if we fix a maximum \emph{churn rate} $\gamma$ (roughly, the fraction of participants online during the last expiration period that are allowed to go offline) and set the maximum failure ratio to be a function of $\gamma$ as depicted in~\Cref{fig:sync-failure-ratio}.
In turn, using votes from a fixed expiration period allows tolerating periods of asynchrony shorter than the expiration period.

Although using messages from multiple rounds allows tolerating bounded asynchrony, it comes at a cost:
Even during synchrony, to ensure safety and progress of the protocol, we must introduce bounds on the fraction of participants that drop offline after participating at some point during the last expiration period.
Otherwise, safety can be violated because a consensus decision may be witnessed by too few participants, compared to the number of participants that have been active during the expiration period, and then overridden in the following rounds.
Progress may also be hampered, as old votes may be too numerous and prevent votes for a new value from reaching the decision threshold. 

In this paper, we address all these challenges and introduce a simple but effective methodology to augment dynamically available protocols, enabling them to withstand bounded periods of asynchrony. The remainder of this paper is as follows: Section~\ref{sec:model} outlines our system model, provides definitions useful to the subsequent sections, and establishes specific conditions concerning the adversary. The core contribution of this work is detailed in Section~\ref{sec:ar-tob}. Here, we present an overview of our proposed methodology, revisit a dynamically available TOB protocol that serves as a practical example for applying our methodology, and then demonstrate the application of our approach, complete with supporting proofs. Discussions on related work are presented in Section~\ref{sec:related}, and conclusions are drawn in Section~\ref{sec:conclusions}.

\begin{figure}[h]
    \caption{Allowable failure ratio $\tilde{\beta}_{2/3}$ to ensure progress during synchrony with an expiration period of $\eta$ rounds and a churn rate of $\gamma$ per $\eta$ rounds. We assume an algorithm using a decision threshold of $1-\beta=\frac{2}{3}$. If participation is static~(i.e.\ $\gamma=0$), the maximum tolerable failure ratio is~$\frac{1}{3}$, and this matches the upper bound for a decision threshold of~$\frac{2}{3}$. At a drop-off rate of $\gamma\geq \frac{1}{3}$, the system may stall even without failures. As we explain in~\Cref{sec:adversary}, in general we must have $\tilde{\beta}=\frac{\beta - \gamma}{\gamma(\beta - 2) + 1}$.}
    \centering
    \scalebox{0.6}{\input{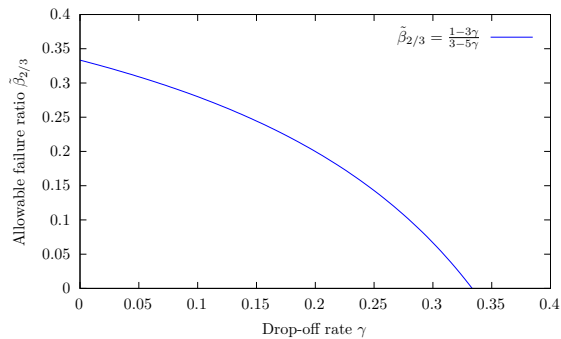}}
    \label{fig:sync-failure-ratio}
\end{figure}

\section{Model and Definitions}
\label{sec:model}

\subsection{System Model}
\label{sec:system-model}

\subsubsection*{Processes}
We consider a system of $n$ \emph{processes} $\mathcal{P} =\{p_1,\ldots, p_n\}$ in a message-passing system with an underlying peer-to-peer dissemination protocol (e.g., a gossip protocol).
Each process is assigned a protocol to follow, consisting of a collection of programs with instructions for all processes.
Processes are divided into \emph{well-behaved} processes and \emph{Byzantine} processes.
Well-behaved processes follow their assigned protocol and send the messages stipulated by it, while Byzantine processes are controlled by an adversary which can make them send arbitrary messages.
Messages sent by processes come with an unforgeable signature, and messages without a valid signature are discarded.

\subsubsection*{Time and Network}
An execution of the system proceeds in an infinite sequence of rounds \(1, 2, 3, \ldots\)
The system is characterized by a known upper bound \(\delta\) that governs the message delay.
Under the assumption that local computations occur instantaneously, it is feasible to establish rounds of duration \(\Delta = 3\delta\)~\cite[Section 2.1]{sleepy}.
Finally, we assume that processes have synchronized local real-time clocks and that, even under asynchrony, clocks remain synchronized.

\subsubsection*{Asynchronous Period}
We assume the existence of a \emph{single asynchronous period} starting \emph{after} round~$r_a$, unknown to the processes, which could extend up to $\pi \in \mathbb{N}$ rounds. In other words, the rounds ranging in $[r_a+1, r_a + \pi]$ may experience asynchrony.

\subsubsection*{Cryptography} We assume a verifiable random function (VRF).  Each process $p_i$ with its secret key can evaluate $(\rho, \pi) \gets \op{VRF}_{p_i}(\mu)$ on any input $\mu$. The output is a deterministic pseudorandom value $\rho$ along with a proof $\pi$. Using $\pi$ and process $p_i$'s public key, anyone can verify whether $\rho$ is the correct evaluation of $\op{VRF}_{p_i}$ on input $\mu$.

\subsubsection*{Sleepiness}
Each round has two phases, one occurring at its beginning and one at its end.
In either phase, only an adversarially chosen subset of the processes are said to be \emph{awake}~\cite{sleepy}.
Processes that are not awake are said to be \emph{asleep}.
The subset of processes awake at the beginning of round~$r$ is $O_r$, and they coincide with the processes awake at the end of the previous round, $r-1$.
In other words, the processes awake at the beginning of a round are potentially different from those awake at the end of it.
Asleep processes do not execute the protocol, and messages for that round are queued and delivered in the first round in which the process is awake again.
When a process~$p_i$ goes from being awake to being asleep, we say that~$p_i$ \emph{goes to sleep}.
We denote with $H_r$ and $B_r$ the sets of well-behaved and Byzantine processes, respectively, that are awake at the beginning of round~$r$.
From now on, we refer to~$H_r$,~$B_r$, and~$O_r$ simply as processes that are awake at round~$r$, leaving it implicit that they are awake at the beginning of it.
The Byzantine processes never go to sleep: the adversary is either \emph{constant}, in which case we have~$B_r$ is the same at every round~$r$, or the adversary is \emph{growing}~\cite{DBLP:journals/iacr/MalkhiMR22}, in which case $B_r\subseteq B_{r+1}$ for every round~$r$.
In this work we will mainly focus on the growing adversary model.

\subsubsection*{Round structure} A round starts with a \emph{send phase}, and ends with a \emph{receive phase}, immediately prior to the beginning of the next round. Processes in $H_r$ participate in the send phase, while processes in $H_{r+1}$ in the receive phase. 

In the send phase of round $r$, each process $p_i \in O_r$ sends messages. A process~$p_i \in B_r$ may send arbitrary messages, and processes that are not awake in round $r$ do not send any messages.
If process~$p_i$ is well-behaved, then~$p_i$ sends the messages dictated by the protocol. 

In the receive phase of round~$r$, each well-behaved process that is awake at the end of round~$r$, i.e., a process in $H_{r+1}$, receives the following messages:
\begin{itemize}
    \item If round~$r$ belongs to a synchronous period, then~$p_i$ receives all the messages that it has not received yet and that were sent in any round $r' \le r$.
    \item Otherwise, if~$r$ belongs to an asynchronous period, then~$p_i$ receives an arbitrary subset of such messages.
\end{itemize}

Moreover, processes that are not awake at the end of round $r$ do not receive any messages. Finally, we assume that messages entering the peer-to-peer messaging protocol are disseminated to all processes, even if the original sender goes to sleep\footnote{Observe that this is roughly the case in most blockchain networks, such as Ethereum (\url{https://ethereum.org}) and Stellar (\url{https://stellar.org}). For example, in Ethereum, process votes are aggregated by intermediate nodes which then disseminate the votes independently.}. Furthermore, these messages withstand the transient asynchronous period we consider and are delivered to all awake processes once normal network conditions are restored.

\subsubsection*{Message structure} We assume that every message exchanged among processes has an \emph{expiration period} \(\eta \in \mathbb{N}\) (except in Section~\ref{sec:mmr}). Specifically, the \emph{expiration period} for round \(r\) is defined as the interval \([r-1 - \eta, r-1]\). Only the messages sent within this interval influence the protocol's behavior at round \(r\). Moreover, each message is tagged with the corresponding round number \(r\) and during each round of a protocol's execution, only the \emph{latest} messages sent by the processes are considered.

\subsection{Total-Order Broadcast and Asynchrony Resilience}

\begin{definition}[Log]
\label{def:log}
A \emph{log} is a finite sequence of blocks $b_i$, represented as $\Lambda = [b_1,b_2,...,b_k]$. Here, a block represents a batch of \emph{transactions} and it contains a reference to another block.
For two logs~$\Lambda$ and~$\Lambda'$, we write $\Lambda \preceq \Lambda'$ when~$\Lambda$ is a prefix of~$\Lambda'$, and we also say that~$\Lambda'$ extends~$\Lambda$.
We say that two logs are \emph{compatible} when one is a prefix of the other, and that they conflict when they are not compatible (i.e., none is a prefix of the other).
\end{definition}

\begin{definition}[{Byzantine total-order broadcast}]
\label{def:tob}
A Byzantine total-order broadcast (TOB) protocol ensures that all the well-behaved processes deliver compatible, growing logs.
In a Byzantine total-order broadcast protocol, every process can \emph{input} a value~$v$ and the broadcast primitive repeatedly \emph{delivers} logs~$\Lambda$.

A protocol for Byzantine total-order broadcast satisfies the following properties.

\begin{description}
    \item[Safety:] If two well-behaved processes deliver logs~$\Lambda_1$ and~$\Lambda_2$, then~$\Lambda_1$ and~$\Lambda_2$ are compatible.

    \item[Liveness:] For every valid transaction\footnote{A transaction is valid according to a global, efficiently computable predicate $P$, known to all processes. The specific details of this predicate are omitted.} there exists a log $\Lambda$ containing it and a round~$r$ such that all well-behaved processes awake for sufficiently long\footnote{The duration ``sufficiently long" varies depending on the protocol. For instance, in the MMR protocol~\cite{DBLP:journals/iacr/MalkhiMR22}, it is \(4\Delta\) in expectation, where \(\Delta\) denotes the network delay.} after~$r$ deliver $\Lambda$.

\end{description}

\end{definition} 

\begin{definition}[Dynamically available total-order broadcast]
\label{def:da-tob}
A protocol for \emph{total-order broadcast} is dynamically available if and only if the protocol satisfies safety and liveness (Definition~\ref{def:tob}) provided that at every round~$r$ it holds that $|B_r| < \beta|O_r|$, for some fixed \emph{failure ratio}~$\beta$.
\end{definition}

Currently implemented dynamically available total-order broadcast share a similar underlying structure: They build upon a \emph{graded agreement} primitive. 

\begin{definition}[Graded agreement]
\label{def:ga}
In a graded agreement protocol\footnote{Graded agreement is a variant of connected consensus~\cite{DBLP:conf/wdag/AttiyaW23}}, each process has an input {log} and, at the end of the protocol, outputs a set of logs with each log assigned a grade bit, such that the following properties are satisfied\footnote{It is important to note that various formulations of graded agreement exist, each possessing distinct properties. For example, the version of graded agreement utilized by Momose and Ren~\cite{DBLP:conf/ccs/Momose022} does not ensure bounded divergence. In contrast, the formulation adopted by D'Amato and Zanolini~\cite{DBLP:journals/corr/abs-2310-11331} assures that outputs of any grade will be unique. The properties of the graded agreement significantly influence the implementation of various dynamically available total-order broadcast protocols. In this work we focus our attention on the graded agreement of~\cite{DBLP:journals/iacr/MalkhiMR22}.}.

\begin{description}
\item[Graded consistency:] If a well-behaved process outputs a log with grade~$1$, then all well-behaved processes output the log with grade~$\ge 0$.

\item[Integrity:] If a well-behaved process outputs a log with any grade, then there exists a well-behaved process that inputs the log.

\item[Validity:] Processes output with grade~$1$ the longest common prefix among well-behaved processes' input logs.

\item[Uniqueness:] If a well-behaved process outputs a log with grade~$1$, then no well-behaved processes outputs any conflicting log with grade~$1$.

\item[Bounded divergence:] Each well-behaved process outputs at most two conflicting logs (with grade~$0$).
\end{description}
\end{definition}

Let  \(D_{r_a}\) be the set of logs decided by well-behaved processes in rounds~$\le {r_a}$.

\begin{definition}[Asynchrony resilience]
\label{def:async-res}
A dynamically available Byzantine total-order broadcast protocol is \emph{asynchrony resilient} if it preserves safety during periods of asynchrony, specifically during the interval \([r_a+1, r_a + \pi]\). This means that during \([r_a+1, r_a + \pi+1]\), no well-behaved process awake at round \(r_a\) decides on a log that conflicts with \(D_{r_a}\). Furthermore, after round \(r_a + \pi+1\), no well-behaved process should decide on a log conflicting with \(D_{r_a}\). A Byzantine total-order broadcast protocol is \emph{\(\pi\)-asynchrony resilient} if it can preserve safety during all asynchrony periods lasting up to~\(\pi\) rounds.
\end{definition}

\begin{definition}[Healing after asynchrony]
A Byzantine total-order broadcast protocol \emph{heals from asynchrony} if it restores its operation according to both safety and liveness after a period of asynchrony. Specifically, this means that after round \(r + k\), where \(r\) is the final round of an asynchronous period and \(k > 0\) is a constant, the following conditions are met:

\begin{description}
    \item[Safety:] For any two rounds \(r'\) and \(r''\) such that \(r', r'' > r + k\), and for any two well-behaved processes \(p_i\) and \(p_j\) awake in these rounds respectively, their logs (\(\Lambda_i^{r'}\) and \(\Lambda_j^{r''}\)) are compatible, i.e., either \(\Lambda_i^{r'} \preceq \Lambda_j^{r''}\) or \(\Lambda_j^{r''} \preceq \Lambda_i^{r'}\).
    \item[Liveness:] For every valid transaction, there exists a log \(\Lambda\) that includes it, and there is a round \(r' > r\) after which all well-behaved processes awake for a sufficient duration deliver \(\Lambda\).
\end{description}
\end{definition}

\subsection{Adversary}
\label{sec:adversary}

In dynamically available total-order broadcast protocols, where the protocol's behavior at each round \( r \) is exclusively influenced by messages received in the immediately preceding round \( r-1 \), it is typically sufficient to impose a threshold for the failure ratio \( \beta \), i.e., $|B_r| < \beta |O_r|$. Commonly, $\beta$ is set to $\frac{1}{2}$~\cite{goldfish, gl, DBLP:conf/ccs/Momose022, DBLP:conf/ccs/MalkhiM023, DBLP:journals/corr/abs-2310-11331}, although some work~\cite{DBLP:journals/iacr/MalkhiMR22} also explored $\beta=\frac{1}{3}$ and $\beta=\frac{1}{4}$. However, these protocols do not guarantee safety during periods of asynchrony, as they are assumed to be synchronous.

In this work, we adopt the approach of extending the time frame of prior rounds considered in message evaluation. This extended span aligns with the expiration duration of messages in our model. Importantly, the protocol's behavior at round~$r$ is influenced only by the \emph{most recent}, i.e., latest, unexpired messages sent by each process (See~\Cref{sec:system-model} for a discussion of message expiration).

However, as we discuss next, incorporating messages from an expanded range of prior rounds may diminish the system's resilience to dynamic participation and adversarial actions within such range. Conversely, this approach enhances tolerance against (bounded) periods of asynchrony.

If there are many processes whose \emph{latest} messages are \emph{unexpired} but which are no longer awake in round~$r$, the protocol's behavior may be adversely affected. This is because the adversary could in principle exploit these latest messages to their advantage, as they are not entirely up to date. We prevent this by bounding the \emph{churn rate} of well-behaved processes, i.e., by requiring that the rate at which awake and well-behaved processes go to sleep is bounded by~$\gamma$ per~$\eta$ rounds. Letting~$H_{s, r} = \bigcup_{s\le r' \le r}H_{r'}$ (with $H_{s} \coloneqq \emptyset$ if $s < 0$) be the set of processes that are awake and well-behaved \emph{at some point} in rounds~$[s,r]$, the requirement is then:
\begin{equation}
\label{eq:churn-bound}
|H_{r-\eta, r-1} \setminus H_r| \leq \gamma |H_{r-\eta, r-1}|
\end{equation}
In other words, at most a fraction~$\gamma$ of the well-behaved processes of the last~$\eta$ rounds are allowed to not be well-behaved processes of the current round~$r$. Besides bounding the churn rate, we also as usual need to bound the failure rate of each round, which we do by requiring a failure rate $\tilde{\beta} \leq \beta$, in particular with $\tilde{\beta} = \frac{\beta - \gamma}{\gamma(\beta - 2) + 1}$:
\begin{equation}
\label{eq:failure-ratio}
|B_{r}| < \tilde{\beta}|O_r|
\end{equation}
Here, $\beta$ is meant to be the failure ratio tolerated by the original dynamically available protocol, which is modified to use unexpired latest messages in order to strengthen its resilience to asynchrony. The failure rate of the modified protocol needs to be appropriately lowered, in particular to $\frac{\beta - \gamma}{\gamma(\beta - 2) + 1}$ if the churn rate is bounded by $\gamma$, to account for the additional power derived from exploiting latest messages of asleep processes.

Observe that, if $\gamma = 0$, our first requirement reduces to $|H_{r-\eta, r-1} \setminus H_r| = 0,$ i.e., awake processes do not go to sleep, so that the model reduces to one without dynamic participation. Moreover, $\tilde{\beta} = \beta$, so Equation~\ref{eq:failure-ratio} simply requires the failure ratio~$\beta$ of the original protocol. In other words, no extra stronger assumption is required under the standard synchronous model with constant participation. Observe also that $\eta = 0$ implies $H_{r-\eta,r-1} = \emptyset$, so that the first requirement does not introduce any restriction, regardless of which $\gamma$ we choose. In other words, fully dynamic participation is allowed. We can in particular let~$\gamma = 0$, meaning that the required failure ratio $\tilde{\beta}$ is once again just $\beta$, recovering the original model. Finally, note that $\gamma$ must be smaller than $\beta$, since otherwise Equation~\ref{eq:failure-ratio} requires $|B_r| < 0$: we cannot allow a fraction~$\beta$ of~$H_{r-\eta, r-1}$ to fall asleep before round~$r$, even if there is no adversary, because then $H_r$ cannot possibly meet a~$1-\beta$ quorum over all unexpired messages (if no more processes wake up).

The $\eta$-sleepy model in the work of D'Amato and Zanolini~\cite{rlmd} deals with the same problem. There, the churn rate is not bound explicitly, and instead a single all-encompassing assumption is made, called the $\eta$-sleepiness condition, the equivalent of which in our framework is\footnote{The original formulation more closely resembles $|B_{r} \cup H_{r-\eta, r-1}\setminus H_r| < \rho|H_r|$, where $\rho = \frac{\beta}{1 - \beta}$. This is equivalent to $|B_{r} \cup H_{r-\eta, r-1}\setminus H_r| < \beta |O_{r-\eta, r}|$, and in turn to $|H_r| > (1-\beta)|O_{r-\eta, r}|$.}:
\begin{equation}
\label{eq;eta-sleepiness}
|H_r| > (1-\beta)|O_{r-\eta, r}|
\end{equation}

\subsubsection*{Bounded asynchrony}

As discussed above, a round~$r$ might belong to the period of asynchrony and, if that's the case, a well-behaved process~$p_i$ might receive in~$r$ an arbitrary subset of the messages sent during such period. It is therefore necessary to forbid the awakening of too many well-behaved processes during asynchronous periods, because the messages they receive upon waking up are adversarially controlled, and thus they can be manipulated into sending messages that jeopardize the safety of decisions made \emph{before the period of asynchrony}. To preserve it, we must prevent the adversary from overwhelming the well-behaved processes which were awake in the last round before asynchrony started, round~$r_a$, either with its own messages or with those of newly awake well-behaved processes, or with corruption, since the adversary can grow. 

Analogously to $H_{s,r}$, we define $O_{s,r} = \bigcup_{s\le r' \le r}O_{r'}$ (with $O_r \coloneqq \emptyset$ if $r < 0$). We require the following conditions to hold whenever analyzing behavior related to asynchrony:
\begin{equation}
\label{eq:async-condition1}
    |H_{r_a} \setminus B_{r}| > (1-\beta)|O_{r-\eta, r}| \quad \forall r\in [r_a+1, r_a +\pi+1]
\end{equation}
\begin{equation}
\label{eq:async-condition2}
    H_{r_a} \subseteq H_{r_a + 1}
\end{equation}
The first condition must hold for all rounds in the period of asynchrony \emph{and for the first synchronous round after it}. For such rounds, we require that the well-behaved processes which were awake in the last synchronous round~$r_a$, \emph{and have not since been corrupted}, sufficiently (meaning, with the usual failure ratio) outnumber all other processes awake in the interval. Intuitively, the processes in $H_{r_a}$ attempt to preserve the safety of decisions made before asynchrony, unless they are corrupted, and they must sufficiently outnumber all other processes in order to do so. The reason why we include round $r_a + \pi + 1$, which is itself synchronous, is that round $r_a + \pi$ being asynchronous means that messages are not guaranteed to be received in its receive phase, and thus that processes in $H_{r_a + \pi + 1}$ still do not necessarily have access to up-to-date messages. The second condition simply requires that all process in $H_{r_a}$ are still awake \emph{at the end of round $r_a$}, so that they participate in the receive phase and in particular obtain messages for the current round, from the other processes in $H_{r_a}$. Knowledge of these messages is what prevents them from ``changing their mind" during the period of asynchrony. {Observe that, by the nature of the asynchronous period, round~$r_a$ is not known in advance by the processes. This implies that the conditions might not always be satisfied. However, as we will demonstrate, if these conditions are met during an asynchronous period of length~$\pi$, then the protocol exhibits resilience to such period.}

\section{Asynchrony-resilient Byzantine total-order broadcast}
\label{sec:ar-tob}

 Recent studies on dynamically available total-order broadcast protocols have explored diverse techniques to resolve consensus in the sleepy model and its variants~\cite{goldfish, rlmd, gl, DBLP:journals/iacr/MalkhiMR22, DBLP:conf/ccs/Momose022, DBLP:conf/sp/NeuTT21, DBLP:journals/corr/abs-2310-11331}. However, these protocols share a common limitation -- they are strictly applicable to synchronous models. This restriction is due to the CAP theorem~\cite{DBLP:journals/sigact/GilbertL02}, which stipulates that no consensus protocol can accommodate dynamic participation and simultaneously tolerate network partitions~\cite{DBLP:conf/sp/NeuTT21}. For this reason, methods to overcome this limitation have been proposed, e.g., the ebb-and-flow family of protocols~\cite{DBLP:conf/sp/NeuTT21}, and are currently being implemented (Ethereum). These involve pairing a dynamically available consensus protocol with a partially synchronous protocol, which provides finality, a concept commonly referred to as safety in the standard consensus literature. However, in scenarios where network partitions or asynchronous periods occur, it becomes challenging to ensure reliable outputs from dynamically available protocols. This is particularly relevant in blockchain contexts, where such conditions could lead to reorganizations of the chain output by these dynamically available protocols. Consequently, our research is driven towards developing mechanisms that improve the resilience of dynamically available protocols against (bounded) periods of asynchrony. Our aim is not to contravene the established impossibility result, but rather to make the output from the protocol more resistant to asynchrony. By doing so, even ebb-and-flow protocols can benefit, as the resulting protocol becomes more robust during periods of asynchrony~\cite{DBLP:conf/esorics/DAmatoZ23}.

However, achieving this resilience is not without drawbacks. As detailed in Section~\ref{sec:model}, our approach first necessitates equipping messages with an expiration period of~$\eta$ rounds. This parameter significantly affects the protocol's tolerance to adversaries and the level of dynamic participation among processes. A high value of $\eta$ could severely limit, or even negate, dynamic participation. Consequently, in practical applications, it is essential to carefully calibrate this parameter. The goal is to ensure that the system maintains dynamic availability while accommodating reasonable periods of network asynchrony.

\paragraph{Enhancing Resilience to Asynchrony in Dynamically Available Protocols}

We briefly outline our proposed mechanism to enhance dynamically available protocols, enabling them to resist bounded periods of asynchrony. The mechanism, simple yet effective, involves some key steps:

\begin{enumerate}
    \item \textbf{Choice of Message Expiration Parameter (\(\eta\))}: The first step involves selecting a suitable message expiration parameter, \(\eta\), which dictates the maximum length of asynchrony the protocol can withstand. This parameter should be calibrated to ensure the system maintains dynamic availability while accommodating reasonable periods of network asynchrony. The specifics of this calibration are beyond the scope of this work.

    \item \textbf{Understanding Adversarial Constraints}: Based on Inequality~\ref{eq:churn-bound} and Inequality~\ref{eq:failure-ratio}, it is crucial to understand the adversarial constraints, particularly how they limit the churn rate of participants in relation to the expiration parameter.

    \item \textbf{Protocol Adjustments}: Different protocols may require distinct adjustments. It is essential to modify each step of the protocol to incorporate the handling of the latest messages from participants in the preceding \(\eta\) rounds. We will detail typical adjustments necessary for these dynamically available TOBs, as they are all share a similar underlying structure. Specifically, these protocols operate in rounds, with a series of rounds forming a view. Within each view, a proposal is made and a decision is reached. To make such decisions, one or multiple instances of {graded agreement} are employed within each view~\cite{DBLP:conf/ccs/Momose022, DBLP:journals/iacr/MalkhiMR22, DBLP:conf/ccs/MalkhiM023, gl, DBLP:journals/corr/abs-2310-11331}. We show how to make these adjustments using the Malkhi, Momose, and Ren protocol~\cite{DBLP:journals/iacr/MalkhiMR22}, chosen for its simplicity and instructional value. While more complex protocols necessitate a more detailed examination, it's important to note that our approach remains applicable to them as well.

    \item \textbf{Proof of Preserved Properties}: The final step involves demonstrating that all essential properties of the protocol are preserved when messages are equipped with expiration periods.
\end{enumerate}

This mechanism aims to provide a framework for enhancing the resilience of dynamically available protocols, contributing to more robust and reliable systems in the face of network asynchrony.

\subsection{Total-Order Broadcast Protocol of Malkhi, Momose, and Ren~\cite{DBLP:journals/iacr/MalkhiMR22}}
\label{sec:mmr}

We recall the total-order broadcast protocol proposed by Malkhi, Momose, and Ren~\cite{DBLP:journals/iacr/MalkhiMR22} and present it within its original framework, specifically the growing adversary model.

Malkhi, Momose, and Ren~\cite{DBLP:journals/iacr/MalkhiMR22} propose a total-order broadcast protocol with a resilience of~$\frac{1}{3}$, and expected termination in 6 rounds, without the assumption of participation stabilization~\cite{DBLP:conf/ccs/Momose022, DBLP:journals/corr/abs-2310-11331}. The authors extended the sleepy model~\cite{sleepy} to allow for a growing number of faulty processes, and developed a simple {graded agreement} protocol with a fault tolerance of $\frac{1}{3}$, upon which their TOB protocol is based.

\begin{figure}[ht!]
    \centering
    \begin{boxedminipage}[t]{\columnwidth}\small
Process $p_i$ runs the following algorithm if awake in one of the two phases of round $r$:
\begin{description}
\item[Beginning of round~$r$ -- Send phase:] Multi-cast [\textsc{vote}, $\Lambda$]$_{p_i}$ where $\Lambda$ is the input log.
\item[End of round~$r$ -- Receive phase:] Tally \textsc{vote} messages and decides the outputs as follows. Let $m$ be the number of \textsc{vote} messages received.
    \begin{enumerate}
        \item For any $\Lambda$ voted by $> \frac{2m}{3}$ processes, output ($\Lambda$, $1$)
        \item For any $\Lambda$ voted by $> \frac{m}{3}$ processes (but $\le \frac{2m}{3}$), output ($\Lambda$, $0$)
    \end{enumerate}
\end{description}
If  $\Lambda'$ extends $\Lambda$, then [\textsc{vote}, $\Lambda'$] counts as a vote for $\Lambda$. Two different \textsc{vote} messages from the same process are ignored.
\end{boxedminipage}
    \caption{Graded Agreement $GA$ - Malkhi, Momose, and Ren~\cite{DBLP:journals/iacr/MalkhiMR22}}
\label{fig:graded-agreement-mmr}
\end{figure}

Figure~\ref{fig:graded-agreement-mmr} describes an instance of the graded agreement protocol of Malkhi, Momose, and Ren~\cite{DBLP:journals/iacr/MalkhiMR22}. As in the original formulation, different processes can be awake in the two phases. Every well-behaved process awake at round~$r$ multi-casts a \textsc{vote} message for a log~$\Lambda$. Then, during the receive phase of round~$r$, every awake process~$p_i$ tallies \textsc{vote} messages for the received logs, counting votes for a log extending~$\Lambda$ as one for~$\Lambda$, and ignoring multiple votes from the same process. If there exists a log~$\Lambda$ that has been voted by more than~$\frac{2}{3}$ (or more than $\frac{1}{3}$) of the processes that~$p_i$ heard from, then~$p_i$ outputs~$\Lambda$ with grade~$1$ (or with grade~$0$). 

\begin{algo}
\vbox{
\small
\begin{numbertabbing}\reset
  xxxx\=xxxx\=xxxx\=xxxx\=xxxx\=xxxx\=MMMMMMMMMMMMMMMMMMM\=\kill
  \textbf{Round~1 of view $v$ ($r = 2v-1$)}\\
    Compute outputs from $GA_{v-1,2}$ \label{}\\
     \textbf{if} $GA_{v-1,2}$ outputs $(\Lambda,1)$ \textbf{then} \label{}\\
    \> \textbf{decide} $\Lambda$ \label{alg:mmr-decide}\\
     $\mathcal{L}_{v-1} \gets \Lambda'$ \label{}
    \`$\Lambda'$ is the longest log s.t. $GA_{v-1,2}$ outputs $(\Lambda',*)$ \label{alg:mmr-set-l}\\
     Start $GA_{v,1}$ with a log in the \textsc{propose} message with the largest \label{}\\
     valid $\op{VRF}(v)$ not conflicting with $\mathcal{L}_{v-1}$\label{}\\
     \\
   \textbf{Round~$2$ of view $v$ ($r = 2v$)}\\
   Compute outputs from $GA_{v,1}$ \label{}\\
 Start $GA_{v,2}$ with the longest $\Lambda$ s.t. $GA_{v,1}$ outputs $(\Lambda,1)$ \label{alg:mmr-input}\\
      $\mathcal{C}_v \gets \mathcal{C}$
    \` $\mathcal{C}$ is the longest log s.t. $GA_{v,1}$ outputs $(\mathcal{C}, *)$ \label{} \\
     Multi-cast [\textsc{propose}, $\Lambda'$, $\op{VRF}_{p_i}(v+1)$]$_{p_i}$\label{}
    \` $\Lambda' := b||\mathcal{C}_v$\label{}

\end{numbertabbing}
}
\caption{Total-order broadcast - Malkhi, Momose, and Ren~\cite{DBLP:journals/iacr/MalkhiMR22} (protocol for $p_i$ in view $v$). Process $p_i$ runs the following algorithm if it is awake at round~$r$. View $0$ lasts~$1$ round, Round~$r=0$. At such round, multi-cast [\textsc{propose}, $\Lambda$, $\op{VRF}_{p_i}(1)$]$_{p_i}$ to propose $\Lambda:=[b_0]$. All later views $v\ge 1$ take two rounds ($r=2v-1, 2v$) and work as follows:}
\label{alg:ab-mmr}
\end{algo}

Malkhi, Momose, and Ren~\cite{DBLP:journals/iacr/MalkhiMR22} implement their total-order broadcast protocol (Algorithm~\ref{alg:ab-mmr}) via two instances of graded agreement. Algorithm~\ref{alg:ab-mmr} is executed in views spanning two rounds each, corresponding to two instances of graded agreement (Figure~\ref{fig:graded-agreement-mmr}). The exception is view~$0$, which requires only a single round. Specifically, at round~$1$ of view~$0$, every awake process~$p_i$ multi-casts a [\textsc{propose}, $\Lambda$, $\op{VRF}_{p_i}(1)$]$_{p_i}$ message, proposing~$\Lambda:=[b_0]$.

Subsequently, at round~$1$ of any other view~$v \ge 1$, each awake and well-behaved process calculates the outputs of $GA_{v-1,2}$, deciding for any log~$\Lambda$ that is output with a grade~$1$. In addition, it sets $\mathcal{L}_{v-1}$ as the longest log $\Lambda'$ for which $GA_{v-1,2}$ generates output at any grade. It then initiates a graded agreement instance~$GA_{v,1}$, inputting a log contained in the \textsc{propose} message with the largest valid $\op{VRF}(v)$, ensuring it doesn't conflict with $\mathcal{L}_{v-1}$.

At round~$2$ of this view, every awake and well-behaved process~$p_i$ computes its outputs from $GA_{v,1}$, and starts a graded agreement instance~$GA_{v,2}$ with the input being the longest log $\Lambda$ that $GA_{v,1}$ outputs with a grade~$1$. Notably, due to the validity property, it's always possible to identify such a~$\Lambda$. Furthermore, process~$p_i$ proposes for view~$v+1$ a block~$b$ extending the longest log $\mathcal{C}_v$ where $GA_{v,1}$ outputs $(\mathcal{C}_v, *)$. This means process~$p_i$ multi-casts a [\textsc{propose}, $\Lambda'$, $\op{VRF}_{p_i}(v+1)$]$_{p_i}$ message with $\Lambda' := b||\mathcal{C}_v$.

The total-order broadcast depicted in Algorithm~\ref{alg:ab-mmr} is not resilient to periods of asynchrony, losing its safety regardless of the duration of the asynchrony. This vulnerability arises because the protocol is designed to operate under synchronous conditions. To understand the implications more concretely, consider the adversary's capabilities during an asynchronous period: In such scenarios, the adversary possesses the ability to dictate the set of messages that any well-behaved process receives. This control extends to influencing the decision-making process of these well-behaved processes. Essentially, by selectively manipulating the message flow, the adversary can steer the decisions of these processes in its advantage, breaking the safety property. For instance, D'Amato and Zanolini~\cite{rlmd} show this for the Goldfish protocol~\cite{goldfish}. An analogous reasoning applies to all the other dynamically available TOB protocols as well, as they are all assumed to be synchronous. 

\subsection{Extended Graded Agreement Protocol}

In order to devise a dynamically available Byzantine total-order protocol with deterministic safety that can effectively handle periods of bounded asynchrony, it becomes essential to \emph{extend} the concept of a graded agreement protocol. This adjustment is crucial in facilitating discussions about the ``messages received in previous rounds''.

Graded agreement, by nature, is a \emph{one-shot} primitive. This means it does not produce a sequence of logs but rather is instantiated with specific inputs and, once it provides output, its execution terminates. In this framework, arguments pertaining to ``unexpired messages from previous rounds'' do not fit.

In the subsequent sections, we demonstrate how to enhance the graded agreement protocol initially presented in Section~\ref{sec:mmr} and we also establish that this improved primitive upholds the properties of graded agreement, as outlined in Section~\ref{sec:model}.

In this section, we elaborate on the extension of the graded agreement protocol initially presented in Figure~\ref{fig:graded-agreement-mmr}. This extended protocol maintains a send and receive phase at round~$r$, with the send phase remaining unchanged.

At the beginning of the protocol, each awake process $p_i$ comes equipped with an initial set of \textsc{vote} messages, denoted as~\(\mathcal{M}_0^i\). These messages originate from a set of processes~\(\mathcal{P}_0\), each supporting a specific log~\(\Lambda\). We require that the cardinality of~\(H_r\) exceeds~\(\frac{2}{3}\) of the cardinality of~\(O_r \cup \mathcal{P}_0\), and each set~\(\mathcal{M}_0^i\) contains a maximum of one message per process. Process~\(p_i\) tallies all the votes it has accumulated from round~\(r\) and discards equivocations. Furthermore, it discards votes in \(\mathcal{M}_0^i\) sent by processes from which~\(p_i\) has received a new \textsc{vote} message in round~\(r\). As a result, by the end of the protocol, process~\(p_i\) holds at most one vote per process in~\(\mathcal{P}_0 \cup O_r\). The \textsc{vote} message from round~\(r\) takes precedence over the initial set of votes \(\mathcal{M}_0^i\). The set of all remaining \textsc{vote} messages, referred to as~\(\mathcal{M}_r^i\), is then employed to output logs with a grade, aligning with the methodology in Figure~\ref{fig:graded-agreement-mmr}. The requirement for grade~$0$ is a quorum of \(\frac{1}{3}\) and for grade~$1$, a quorum of~\(\frac{2}{3}\). It is worth noting that when \(\mathcal{M}_0^i = \emptyset\) for all \(p_i\), we revert to the standard graded agreement from Figure~\ref{fig:graded-agreement-mmr}.
\begin{figure}[ht!]
    \centering
    \begin{boxedminipage}[t]{\columnwidth}\small
Process $p_i$ runs the following algorithm if awake in one of the two rounds:\\

    \textbf{Beginning of round~$r$ -- Send phase:} Multi-cast [\textsc{vote}, $\Lambda'$]$_{p_i}$ where $\Lambda'$ is the input log extending~$\Lambda$.\\
    
    \textbf{End of round~$r$ -- Receive phase:} Let $\mathcal{M}_r^i$ be the set of \textsc{vote} messages from round~$r$ and from $\mathcal{M}_0^i$ (discarding equivocations in either set, and discarding messages in $\mathcal{M}_0^i$ if the processes that sent such messages also sent a message in round~$r$). 
    
    Tally \textsc{vote} messages in $\mathcal{M}_r^i$ and decide the outputs as follows. Let $m$ be the number of \textsc{vote} messages received.
    \begin{enumerate}
        \item For any $\Lambda$ voted by $> \frac{2m}{3}$ processes, output ($\Lambda$, $1$)
     \item For any $\Lambda$ voted by $> \frac{m}{3}$ processes (but $\le \frac{2m}{3}$), output ($\Lambda$, $0$)
    \end{enumerate}
If  $\Lambda'$ extends $\Lambda$, then [\textsc{vote}, $\Lambda'$] counts as a vote for $\Lambda$. Two different \textsc{vote} messages from the same process are ignored.
\end{boxedminipage}
    \caption{Extended Graded Agreement $GA$ initialized with a set $\mathcal{M}_0^i$ of \textsc{vote} messages from a set of processes~$\mathcal{P}_0$, each supporting some log~$\Lambda$ -- protocol for process~$p_i$.}
\label{fig:graded-agreement-mmr-ext}
\end{figure}

\begin{lemma}
\label{lem:extended-ga-prop}
   The extended   graded agreement presented in Figure~\ref{fig:graded-agreement-mmr-ext} satisfies the original properties of graded agreement (Section~\ref{sec:model}). It moreover satisfies the following property, both for  synchronous and asynchronous rounds.

\begin{description}
    \item[Clique validity:] Consider $H' \subset H_r \cup H_{r+1}$ such that all $p_i \in H' \cap H_r$ have an extension of $\Lambda$ as input, and such that, for any $p_i \in H' \cap H_{r+1}$, $\mathcal{M}_0^i$
    contains a message from each process in $H'$, also all for some extension of $\Lambda$. Moreover, suppose that $|H'| > \frac{2}{3} |O_r \cup \mathcal{P}_0|$. Then, all processes in $H' \cap H_{r+1}$ output $\Lambda$ with grade 1.
\end{description}
    
\end{lemma}

\begin{proof}
    The proofs of the shared properties is similar as in the original protocol. There, we use that $|H_r| > \frac{2}{3}|O_r|$, whereas here we use $|H_r| > \frac{2}{3}|O_r \cup \mathcal{P}_0|$ in an analogous manner, as $O_r \cup \mathcal{P}_0$ is set the of all processes whose messages can influence the outputs, much like $O_r$ in the original protocol. 

Let us consider a round~$r$, and let $n_r$ be the maximum possible perceived participation by any well-behaved participant awake in round~$r$, i.e., $n_r = \left|\mathcal{P}_0 \cup O_r\right|$. We repeatedly use the assumption that $|H_r| > \frac{2}{3}n_r$. Moreover, for all properties other than clique validity, network synchrony is assumed, so we repeatedly use that, for all $p_i \in H_{r+1}$, $H_r \subset \mathcal{M}_r^i$, since all well-behaved messages from $H_r$ are broadcast on time and thus received by the end of the round.

For the \emph{graded consistency} property, let us assume that process~$p_i$ outputs a log~$\Lambda$ with grade~$1$ and let $m = |\mathcal{M}_r^i|\leq n_r$ be the perceived participation of process~$p_i$. Moreover, let~$S$ be the set of processes whose message in $\mathcal{M}_r^i$ is for an extension of $\Lambda$. By assumption, $|S| > \frac{2}{3}m$, and $|H_r| > \frac{2}{3}n_r$. Moreover, $|S| + |H_r| - |S \cap H_r| = |S \cup H_r| \leq m$, since $S, H_r \subset \mathcal{M}_r^i$. Therefore, $|S \cap H_r| \geq |S| + |H_r| - m > \frac{2}{3}(n_r + m) - m = \frac{2}{3}n_r - \frac{m}{3} \geq \frac{2}{3}n_r - \frac{n_r}{3} = \frac{n_r}{3}$, i.e., $|S \cap H_r| > \frac{n_r}{3}$. For any process $p_j \in H_{r+1}$, $S \cap H_r \subset \mathcal{M}_r^j$, so $p_j$ counts $ > \frac{n_r}{3}$ votes for extensions of~$\Lambda$, and it thus outputs~$\Lambda$ with at least grade~$0$.

The proof for the \emph{integrity} property follows from a very similar argument as for graded consistency, in this case with $|S| > \frac{m}{3}$. In particular, $|S \cap H_r| \geq |S| + |H_r| - m > \frac{m}{3} + \frac{2}{3}n_r - m = \frac{2}{3}(n_r - m)$. Since $m\le n_r$, it follows that $S \cap H_r \neq \emptyset$, implying that at least a well-behaved process voted for a log extending~$\Lambda$.

For \emph{validity}, let~$\Lambda$ be the longest common prefix among well-behaved processes' inputs logs at round~$r$. Every process in~$H_r$ multi-casts a \textsc{vote} message for an extension of~$\Lambda$. The proof easily follows from the assumption that~$|H_r|>\frac{2}{3}n_r$, and from $H_r \subset \mathcal{M}_r^i$ for all $p_i \in H_{r+1}$. 

To prove \emph{uniqueness}, let us assume that a well-behaved process~$p_i$ awake at round~$r$ outputs a log~$\Lambda$ with grade~$1$. By the same logic of the graded consistency property, we have that every other well-behaved process~$p_j$ awake at round~$r$ sees $|S \cap H_r| > \frac{n_r}{3}$ \textsc{vote} messages for an extension of~$\Lambda$. This implies that there cannot be a well-behaved process~$p_j$ that sees more that $\frac{2}{3}m$ \textsc{vote} messages for a conflicting log.

For \emph{bounded divergence}, observe that in order to be output with any grade by process $p_i$, a log~$\Lambda$ must be voted by more than $\frac{m}{3}$ processes, where $m = |\mathcal{M}_r^i|$ is the perceived participation of $p_i$. Recall that $\mathcal{M}_r^i$ contains at most one message per process. Thus, each process outputs at most two conflicting logs.

Finally, for the \emph{clique validity} property, let us consider a process $p_i \in H' \cap H_{r+1}$. By assumption, $\mathcal{M}_0^i$ contains a \textsc{vote} message for some extension of~$\Lambda$ from each process in $H'$. Since all \textsc{vote} messages from $H' \cap H_r$ are also by assumption for some extension of~$\Lambda$, it is the case that $\mathcal{M}_r^i$ also contains a \textsc{vote} message for each process in~$H'$, all for extensions of~$\Lambda$. By assumption, $|H'| > \frac{2}{3} |O_r \cup \mathcal{P}_0| = \frac{2}{3} n_r$ which implies that~$\Lambda$ is output with grade 1 by~$p_i$.
\end{proof}

Note that the mechanism of extending the graded agreement protocol by Malkhi, Momose, and Ren~\cite{DBLP:journals/iacr/MalkhiMR22} is also applicable to other graded agreement protocols. Specifically, by providing each process \( p_i \) with a set \( \mathcal{M}^i_0 \) possessing the characteristics defined above, the same logic can be applied to adapt other protocols similarly.

\subsection{Extended Byzantine Total-Order Broadcast Protocol}

We show that the extended graded agreement protocol from the previous section can be used to capture expiration of messages in the~$\eta$-sleepy model, allowing us to simply prove safety and liveness of Algorithm~\ref{alg:ab-mmr} in the~$\eta$-sleepy model with messages subject to expiration.

{Recall that Algorithm~\ref{alg:ab-mmr} proceeds in views of two rounds each, and in each round an instance of graded agreement (Figure~\ref{fig:graded-agreement-mmr}) is started. In order to make Algorithm~\ref{alg:ab-mmr} asynchrony resilient, we modify it to use the latest unexpired messages as inputs in its graded agreement instances, i.e., a process $p_i \in H_{r+1}$ computes its outputs from a $GA$ instance started in round~$r$ based on the set of unexpired, latest messages, i.e., the latest among those from rounds $[r -\eta, r]$, with equivocating latest messages being discarded. {From this point, ``Algorithm~\ref{alg:ab-mmr} modified to use latest unexpired messages'', or simply ``the modified Algorithm~\ref{alg:ab-mmr}'',  refers precisely to this modified protocol.}

Note that a $GA$ instance in the modified Algorithm~\ref{alg:ab-mmr} corresponds exactly to a specific instance of the extended graded agreement depicted in Figure~\ref{fig:graded-agreement-mmr-ext}. In particular, the $GA$ instance at round~$r$ of the modified Algorithm~\ref{alg:ab-mmr} corresponds to an instance of extended graded agreement protocol where the initial set $\mathcal{M}^i_0$ of process $p_i$ is taken to contain the set of all latest messages among those from rounds $[r-\eta, r)$ seen by $p_i$, with equivocating latest messages being discarded. The set $\mathcal{M}^i_{r}$, which $p_i \in H_{r+1}$ uses to determine its output in the extended graded agreement protocol used in Algorithm~\ref{alg:ab-mmr}, contains simply all latest unexpired messages, i.e., the latest messages among those from rounds $[r -\eta, r]$ (without equivocations). It is helpful to show that the assumptions of the extended graded agreement protocol hold when instantiated in the context of the modified Algorithm~\ref{alg:ab-mmr}.
For an extended graded agreement protocol happening at a synchronous round~$r$, we have required that $|H_r| > \frac{2}{3}|O_r \cup \mathcal{P}_0|$. In the particular instance we have constructed above in the context of Algorithm~\ref{alg:ab-mmr}, $\mathcal{P}_0$, the set of senders of messages in $\mathcal{M}_0^i$, is contained in $H_{r-\eta, r-1} \cup B_r \subseteq O_{r-\eta, r}$, since by construction $\mathcal{M}_0^i$ contains only messages from rounds $[r-\eta, r)$, and we are at round~$r$. Therefore, $|O_r \cup \mathcal{P}_0| \le |O_{r-\eta, r}|$, and thus $|H_r| > \frac{2}{3}|O_r \cup \mathcal{P}_0|$ immediately follows from the $\eta$-sleepiness assumption $|H_r| > \frac{2}{3}|O_{r-\eta, r}|$.  Since all the assumptions hold, Lemma~\ref{lem:extended-ga-prop} guarantees that graded consistency, integrity, validity, uniqueness, and bounded divergence all apply to the extended   graded agreement instances used in our modified Algorithm~\ref{alg:ab-mmr}. 

\begin{theorem}
\label{thm:tob}
    Algorithm~\ref{alg:ab-mmr} with the extended graded agreement protocol implements Byzantine total-order broadcast. 
\end{theorem}

\begin{proof}
    As we have just discussed, each instance of the extended graded agreement protocol utilized in the modified Algorithm~\ref{alg:ab-mmr} satisfies the five properties of the graded agreement primitive from Malkhi, Momose, and Ren~\cite{DBLP:journals/iacr/MalkhiMR22}. Since the safety and liveness proofs of Algorithm~\ref{alg:ab-mmr} (Lemma 6 and Lemma 7 of~\cite{DBLP:journals/iacr/MalkhiMR22}) rely entirely on these properties, they apply to the modified Algorithm~\ref{alg:ab-mmr} as well.
\end{proof}

Recall that $r_a$ is the last round before asynchrony starts. We have the following results for Algorithm~\ref{alg:ab-mmr} modified to use latest messages.

\begin{lemma}
\label{lem:clique-sticks-together}
    Let $[r_a + 1, r_a + \pi]$ with $\pi < \eta$ be the period of asynchrony. If every process~$p_i$ in $H_{r_a}$ multi-casts a \textsc{vote} message for an extension of a log~$\Lambda$ in round~$r_a$, then every process~$p_i \in H_{r_a} \cap H_r$ multi-casts a \textsc{vote} message for an extension of log~$\Lambda$, for every round~$r \in [r_a + 1, r_a + \pi+1]$.
\end{lemma}

\begin{proof}
We prove this lemma through an inductive argument. The base case is round~$r_a + 1$. By assumption (Equation~\ref{eq:async-condition2}) we have that $H_{r_a} \subseteq H_{r_a+1}$, i.e., processes participating in the send phase of the extended graded agreement of round~$r_a$ also participate in its receive phase. In particular, seen that round~$r_a$ is synchronous by assumption, each process~$p_i \in H_{r_a}$ receives all the \textsc{vote} messages for an extension of~$\Lambda$ sent by other processes in~$H_{r_a}$ in round~$r_a$. By validity property of the extended graded agreement, every process in~$H_{r_a+1}$ outputs from GA the log~$\Lambda$ with grade~$1$, and thus multi-casts a \textsc{vote} message for an extension of it in the next instance of the extended graded agreement of round~$r_a + 1$.

For the inductive step, suppose that every process~$p_i \in H_{r_a} \cap H_r$ multi-casts a \textsc{vote} message for an extension of~$\Lambda$, for every round~$r \in [r_a + 1, r']$, with $r' < r_a + \pi + 1$. Let $H' = H_{r_a} \setminus B_{r'}$, and observe that for every $p_i \in H' \cap H_{r'+1}$, the set~$\mathcal{M}_0^i$ contains all latest unexpired \textsc{vote} messages from rounds~$< r'$ that~$p_i$ has received. In particular it contains a latest and unexpired message from each process in $H_{r_a} \setminus B_{r'}$, all from rounds no later than $r_a$. This is because messages from round~$r_a$ from $H_{r_a}$ were previously received in round $r_a$, and these are still unexpired, since $r'+1 - \eta \leq r_a + \pi + 1 - \eta < r_a$. By inductive assumptions, all such latest messages are for an extension of~$\Lambda$. It is then the case that all $p_i \in H' \cap H_{r'}$ have an extension of $\Lambda$ as input, and that, for any $p_i \in H' \cap H_{r'+1}$, $\mathcal{M}_0^i$ contains a message from each process in $H'$, also all for some extension of $\Lambda$, as required by the assumptions of clique validity. To apply clique validity, we need only to show that $|H'| > \frac{2}{3}|S_{r'} \cup P_0|$. Equation~\ref{eq:async-condition1} gives us that $|H'| = |H_{r_a} \setminus B_{r'}| > \frac{2}{3}|O_{r'-\eta, r'}|$, which immediately implies the desired result, because by construction $\mathcal{P}_0 \subset O_{r'-\eta, r'}$, so $|O_{r'} \cup \mathcal{P}_0| \le |O_{r'-\eta, r'}|$.  
\end{proof}

The following lemma describes the behavior of MMR under synchrony, which is preserved when modifying Algorithm~\ref{alg:ab-mmr} to use latest unexpired messages, as we have already argued. This result will then be used in the proof of Theorem~\ref{thm:async-res}.

\begin{lemma}
\label{lem:decided-implies-voted}
Let~$\Lambda \in D_{r_a}$ be a log decided in a round $r \le r_a$. In every round~$r' \in [r, r_a]$, every process~$p_i \in H_{r'}$ multi-casts a \textsc{vote} message for an extension of log~$\Lambda$. 
\end{lemma}

\begin{proof}
Let~$\Lambda \in D_{r_a}$ be a log decided in round $r \le r_a$ by process~$p_i$ awake at round~$r$, i.e., $p_i$ has outputs $\Lambda$ with grade~$1$ in round~$r$. We prove this result through induction on rounds $r' \in [r, r_a]$. 

The base case, i.e., $r' = r$, follows from the graded consistency property of the extended graded agreement. In particular, if an awake and well-behaved process~$p_i$ decides~$\Lambda$ in round~$r \le r_a$, then, since $r$ is a synchronous round, all processes $p_i \in H_r$ multi-cast a \textsc{vote} message for an extension of~$\Lambda$.

For the induction step, suppose that if every process in $H_{r'}$ multi-casts a \textsc{vote} message for an extension of~$\Lambda$, then (from the validity property of the extended graded agreement) every process in $H_{r'+1}$ outputs~$\Lambda$ with grade $1$. This implies that then they all multi-cast a \textsc{vote} message for an extension of~$\Lambda$.
\end{proof}

\begin{theorem}
\label{thm:async-res}
Algorithm~\ref{alg:ab-mmr} with the extended   graded agreement protocol is~$\pi$-asynchrony resilient for $\pi < \eta$.
\end{theorem}

\begin{proof}
Let~$\Lambda \in D_{r_a}$ be a log decided in round $r \le r_a$ and let $[r_a + 1, r_a + \pi]$ with $\pi < \eta$ be the period of asynchrony. By Lemma~\ref{lem:decided-implies-voted}, all processes in $H_{r_a}$ multi-cast a \textsc{vote} message for an extension of~$\Lambda$ in rounds~$[r, r_a]$. In particular they do so in round~$r_a$, so we can apply Lemma~\ref{lem:clique-sticks-together} and conclude that every process in $H_{r_a} \cap H_{r'}$ also multi-casts a \textsc{vote} message for an extension of~$\Lambda$ in round~$r'$, for any round  $r' \in [r_a+1, r_a + \pi+1]$. Firstly, this shows that no process $p_i \in H_{r_a}$ ever decides a log~$\Lambda'$ conflicting with~$\Lambda$ in rounds $[r, r_a + \pi + 1]$, as this would imply multi-casting a \textsc{vote} message for an extension of~$\Lambda'$. Moreover, since round~$r_a+\pi+1$ is synchronous by assumption, all \textsc{vote} messages from rounds $[r_a, r_a + \pi + 1]$ are delivered in the receive phase of the round to all well-behaved processes which are awake during it, i.e., to processes in $H_{r_a+\pi+2}$. Any such process would then have received all messages sent by processes in $H_{r_a}$ in rounds~$[r_a, r_a+\pi+1]$, which are all unexpired at round~$r_a+\pi+2$, since the expiration period for it starts at round $(r_a+\pi+2) - 1 - \eta = r_a + 1 + \pi - 
\eta \le r_a$ since $\pi < \eta$. Therefore, any process $p_i \in H_{r_a+\pi+2}$ has received an unexpired message for each process in $H_{r_a} \setminus B_{r_a + \pi + 1}$, all for extensions of~$\Lambda$, since no other messages were cast during rounds~$[r_a, r_a+\pi+1]$ by such processes. In particular, the latest of these messages is then for an extension of~$\Lambda$. Equation~\ref{eq:async-condition1} then gives us $|H_{r_a} \setminus H_{r_a+\pi+1}| > \frac{2}{3} |O_{r_a+\pi+1-\eta, r_a+\pi+1}|$ $> \frac{2}{3}$ of all latest unexpired messages seen by $p_i$ are for an extension of~$\Lambda$, and thus $p_i$ outputs~$\Lambda$ with grade 1 and multi-casts a \textsc{vote} message for an extension of it in round~$r_a+\pi+2$. Since rounds $\ge r_a+\pi+2$ are synchronous, we can then apply the same inductive reasoning of the original protocol (MMR) and conclude that all processes $H_{r'}$ multi-cast a \textsc{vote} message for an extension of~$\Lambda$ in any round $r' \ge r_a + \pi + 2$. In particular, this rules out any decision for a conflicting log in such rounds. Overall, we have shown that processes in $H_{r_a}$ never decide a log conflicting with~$D_{r_a}$, and after round $r_a + \pi+1$ no well-behaved process at all decides a log conflicting with~$D_{r_a}$, i.e., that the protocol is $\pi$-asynchrony-resilient.
\end{proof}

\begin{theorem}
Algorithm~\ref{alg:ab-mmr} with the extended graded agreement protocol heals after any period of asynchrony after $k=1$ slots.
\end{theorem}

\begin{proof} 
Let~$r$ be the first round after asynchrony, and view~$v$ be any view whose first round is~$\ge r$. By assumption, $\eta$-sleepiness holds at all rounds of views $\ge v$, so all such rounds satisfy the graded agreement properties. Thus, all decisions made in views $\ge v$ are safe, and all proposals from well-behaved proposers made in such views have a probability~$\frac{1}{2}$ of being decided. In other words, the protocol is safe and live after round~$r$.
\end{proof}

\section{Related work}
\label{sec:related}

Pass and Shi's ``Sleepy Model"~\cite{sleepy} marked a key formalization in distributed protocols. Their work presents a significant shift in consensus protocols, formalizing the concept, initially adopted by the Bitcoin~\cite{nakamoto2008bitcoin} protocol, of participants fluctuating between being online or offline during a protocol execution. 

Momose and Ren~\cite{DBLP:conf/ccs/Momose022} present a total-order broadcast protocol that supports dynamic participation while achieving constant latency. The authors do this by extending the classic Byzantine Fault Tolerance (BFT) approach from a static quorum size to a dynamic one, adjusting according to the current level of participation. 

Another stride towards accommodating fluctuating participation was made by Malkhi, Momose, and Ren~\cite{DBLP:journals/iacr/MalkhiMR22}. This work presents a protocol with a significantly reduced latency of three rounds, which tolerates one-third malicious participants and allows fully dynamic participation of both well-behaved and malicious participants. Subsequently, Malkhi, Momose, and Ren~\cite{DBLP:conf/ccs/MalkhiM023} improve on their previous work~\cite{DBLP:journals/iacr/MalkhiMR22} by providing a dynamically available Byzantine total-order broadcast protocol under dynamic and unknown participation with an assumption of minority corruption. 

Gafni and Losa~\cite{gl} present two consensus algorithms that tolerate a ratio of~$\frac{1}{2}$ malicious failures in the the sleepy model. The first algorithm achieves deterministic safety and probabilistic liveness with constant expected latency, while the second, albeit theoretically due to its high round and message complexity, offers deterministic safety and liveness.

D'Amato and Zanolini~\cite{DBLP:journals/corr/abs-2310-11331} propose a total-order broadcast protocol in the sleepy model, designed to withstand adversarial behavior from up to 50\% of the participants. A significant advancement of their protocol is its efficiency, necessitating just one voting round for each decision, in contrast to earlier protocols that demanded multiple voting rounds per decision.

Focusing on Ethereum's consensus protocol, D'Amato and Zanolini~\cite{rlmd} tackle the challenge of tolerating periods of asynchrony in LMD-GHOST, the dynamically available component of Gasper~\cite{gasper}. The authors present RLMD-GHOST, a synchronous consensus protocol that not only ensures dynamic availability but also maintains safety during bounded periods of asynchrony. Unlike our result, which concentrates on deterministically safe, dynamically available consensus protocols, D'Amato and Zanolini~\cite{rlmd} focus on a consensus protocol that is probabilistically safe and dynamically available.

\section{Conclusions}
\label{sec:conclusions}

In this work we studied the problem of handling asynchrony in dynamically available protocols that are \emph{deterministically safe}. Our main contribution revolves around the concept of a configurable message-expiration period applied to the 1/3-resilient, dynamically available total-order broadcast protocol of Malkhi, Momose, and Ren~\cite{DBLP:journals/iacr/MalkhiMR22}. By leveraging the latest votes of participants across multiple prior rounds instead of restricting to the current round, we presented an effective mechanism to enhance the resilience of the protocol during asynchrony. In order to benefit from this approach, we introduced a ``churn rate" to quantify the maximal fraction of online participants that can transition to an offline state. We have shown that this churn rate plays a central role in determining the maximum tolerable failure ratio during synchronous operations. The techniques utilized in this work can also be directly applied to other deterministically safe, dynamically available protocols, and we leave an in-depth analysis of this for future work.

\section*{Acknowledgments}
The authors thank anonymous reviewers for interesting discussions and helpful feedback.

\bibliographystyle{plainurl}
\bibliography{references}

\end{document}